\newif\ifLINUXBUILD

\IfFileExists{/dev/null}{\LINUXBUILDtrue}{}

\documentclass[runningheads]{llncs}
\usepackage{xcolor}
\usepackage[pdftex,colorlinks=true,pdfstartview=FitV,linkcolor=black,citecolor=blue,urlcolor=black]{hyperref}

\usepackage[square,comma,numbers,sectionbib]{natbib}
\usepackage{amsfonts}
\usepackage{amsmath}
\usepackage{mathtools}

\newtheorem{myAttack}{Attack}

\makeatletter
\renewcommand\@biblabel[1]{#1.}
\makeatother
\renewcommand\orcidID[1]{}
\author{Ignacio~M.~Delgado-Lozano\inst{1}\orcidID{0000-0002-0003-3318} \and
Macarena~C.~Mart\'inez-Rodr\'iguez\inst{2}\orcidID{0000-0003-3025-5736} \and
Alexandros~Bakas\inst{1}\orcidID{0000-0002-0731-1851} \and
Billy~Bob~Brumley\inst{1}\orcidID{0000-0001-9160-0463} \and
Antonis~Michalas\inst{1}\orcidID{0000-0002-0189-3520}}
\institute{Tampere University, Tampere, Finland\\
\email{\{ignacio.delgadolozano,alexandros.bakas\}@tuni.fi} \and
Instituto de Microelectr\'onica de Sevilla,\\CSIC/Universidad de Sevilla, Sevilla, Spain\\
\email{macarena@imse-cnm.csic.es}}
\authorrunning{Ignacio~M.~Delgado-Lozano et al.}
\index{Delgado-Lozano, Ignacio M.}
\index{Mart\'inez-Rodr\'iguez, Macarena C.}
\index{Bakas, Alexandros}
\index{Brumley, Billy Bob}
\index{Michalas, Antonis}

\usepackage[T1]{fontenc}
\usepackage{graphicx}
\usepackage{subfigure}
\usepackage{amsmath}
\usepackage{xspace}
\usepackage{lipsum}
\usepackage{booktabs}

\newcommand{\Paragraph}[1]{\medbreak\noindent\textbf{#1.}}
\newcommand{\KEYWORDS}{%
attestation \and
remote power analysis \and
side channels \and
ADC \and
secure protocols \and
secure communications
}

\title{Attestation~Waves: Platform~Trust~via~Remote~Power~Analysis}

\begin{document}

\maketitle

\begin{abstract}

Attestation is a strong tool to verify the integrity of an untrusted system.
However, in recent years, different
attacks have appeared that are able to mislead the attestation process
with treacherous practices as memory copy, proxy, and rootkit attacks, just to name a few.
A successful attack leads to systems that are considered trusted
by a verifier system, while the prover has bypassed the
challenge. To mitigate these attacks against attestation
methods and protocols, some proposals have considered the use
of side-channel information that can be measured externally,
as it is the case of electromagnetic (EM) emanation. Nonetheless,
these methods require the physical proximity of an external setup
to capture the EM radiation.

In this paper, we present the possibility
of performing attestation by using the side-channel information
captured by a sensor or peripheral that lives in the same
System-on-Chip (SoC) than the processor system (PS) which
executes the operation that we aim to attest, by only sharing
the Power Distribution Network (PDN). In our case, an
analog-to-digital converter (ADC) that captures the voltage
fluctuations at its input terminal while a certain operation
is taking place is suitable to characterize itself and to distinguish it
from other binaries. The resultant power traces are enough to clearly identify
a given operation without the requirement of physical proximity.
\end{abstract}

\section{Introduction} \label{sec:intro}

In our current network and interconnected world, establishing platform trust
for execution of different security-critical operations is a need in diverse
fields: examples include manufacturing, automation, communications, transport,
work, and finance \cite{DBLP:journals/ijisec/CokerGLHMORSSS11}.
One approach is to use attestation mechanisms, which are suitable to verify the
integrity of several elements such as application binaries, data, or other
internal platform state. Attestation normally consists of presenting a challenge
by a verifier system that is already trusted to a prover system.

Attestation is a powerful concept to verify the integrity of untrusted systems.
Recently, different attacks have appeared that aid in circumventing attestation
by making a copy of the code that generates the checksum expected by the
verifier (memory copy attack) \cite{Seshadri2005,Seshadri2006,Chen2017}, forwarding the challenge to another device that
is able to compute the checksum properly (proxy attack) \cite{Li2011}, or using return
oriented programming gadgets to transiently hide the malicious code in parts of
memory where the verifier cannot find it (rootkit attack) \cite{Castelluccia2009}. As a result, what we
get are systems that are considered trusted by a verifier system, while the
prover has bypassed the challenge.

To harden against these attacks on attestation methods and protocols, some
proposals have considered the use of side-channel information that can be
measured externally. For example, \citet{DBLP:conf/micro/SehatbakhshNKZP19}
recently utilized electromagnetic (EM) emanation to verify honest checksum
computation. Nonetheless, these methods require proximity: a local external
testbed set up near the prover in such a way that a carefully-placed probe can
capture the EM radiation (traces) of the prover's device. Furthermore, this
testbed itself must be secured and trusted. The physical proximity requirement
directly contradicts with the goals of remote attestation, not to mention
failure to scale.

Recent trends in offensive cryptanalytic side-channel analysis are towards
\emph{remote power analysis} \cite{DBLP:conf/IEEEares/Martinez-Rodriguez21}.
These techniques allow attackers to utilize
pre-existing sensors or peripherals living in the same System-on-Chip (SoC) to
procure traces. Regarding cryptanalytic side-channel attacks, this removes the
physical proximity requirement from the threat model. In practice, these traces
feature granularity reduced by several orders of magnitude when compared to
traces captured with traditional high sampling rate oscilloscopes
(e.g.\ 1MSPS in \autoref{sec:measure} vs.\ 40GSPS in \cite{DBLP:journals/tches/LisovetsKMM21}).
Hence, remote power analysis trades this relaxed threat model for lower quality
and higher quantity of traces. \autoref{sec:bg} contains more background on both
remote power analysis and attestation.

In this paper, we propose utilizing remote power analysis for remote dynamic
attestation, eliminating the physical proximity requirement of previous EM-based
attestation proposals. \autoref{sec:measure} describes our testbed, with an
application processor (AP) that executes the binary we aim to attest, by only
sharing the Power Distribution Network (PDN) with the sensor that captures the traces.
In our case, an analog-to-digital
converter (ADC) captures the voltage fluctuations at its input terminal during
attestation. \autoref{sec:protocol} proposes an attestation protocol to establish
secure communication between prover and verifier systems in a platform-agnostic way.
\autoref{sec:eval} characterizes the degree to which the
resulting traces captured from ADC vary over different binaries, with the goal of accurately
matching traces to a priori applications with signal processing techniques via
templating.
In particular, we show that with a sufficient (yet small) number of traces, parameterized (in part)
by various error rates, we are able to achieve excellent security levels and also understand
the limitations of attestation in this novel setting.
We conclude in \autoref{sec:conclusion}.
\section{Background} \label{sec:bg}

In a typical software-based attestation, a verifier is able to establish the
absence of malware in a prover system with no physical access to its
memory. This is possible because the verifier proposes a challenge to the prover,
in which it must compute a checksum of its memory content.
This challenge can only be correctly replied to if the memory content is not
tampered, since the result of the checksum is only correct if the memory
content within the prover system is exactly as expected by the verifier.
For this, the verifier system needs to know several critical data about the
prover, such as the clock speed, the instruction set architecture, the memory
architecture of its microcontroller, and the size of its memories. This way,
if in any moment a malicious prover aims to alter its memory, it is
detectable by the verifier because the prover will present a wrong
checksum result or a delay in the response \cite{Seshadri2004}.
This means that the integrity of the prover is verified, not only
matching the checksum result with the expected result
$(Response_{prover} = Response_{expected})$, but also through a
parameter known as the request-to-response time $(t_{response} < t_{expected})$.

Numerous works focus on software-based attestation
\cite{Seshadri2005,Seshadri2006,Chen2017}.
To threaten these attestation processes, several attacks have appeared
during the last fifteen years that aim to break this attestation method.
Attackers normally attempt to forge the response
with a checksum computed in a different region of the
prover memory that duplicates the code. This allows them to
generate the expected response, which is known as a
memory copy attack \cite{Seshadri2005,Seshadri2006,Chen2017}.
Another possibility is to forward the challenge to another device
that is able to compute the checksum, then send it back to the
verifier while satisfying the request-to-response time requirement, leading
to proxy attacks. \citet{Li2011} extensively describe proxy attacks
and present attestation protocols to prevent them. The last option
consists of storing the malicious code previous to the checksum
calculation, by hiding it in other parts of the memory, allowing
the prover to compute the correct checksum while the verifier is not
able to detect the parts of the code that have been hidden, called
a rootkit attack \cite{Castelluccia2009}.

Along with software-based attestation processes, hardware-supported Trusted Execution
Environments (TEE) are frequently used to ensure that the response of the computation
is not tampered. \citet{Abera2016} allow remote control-flow path attestation of
an application without needing the code. They utilize ARM's TrustZone (TZ) in order
to avoid memory corruption attacks. \citet{Clercq2016} present a Control Flow
Integrity (CFI) mechanism, guarding against code injection and code reuse attacks.
They also present Software Integrity (SI) by storing
precomputed MACs of instructions and comparing them with the MAC
of the run-time execution. Moreover, \citet{Dessouky2017} monitor every branch,
a mitigation leveraging un-instrumented
control-flow instructions.

Besides classical attestation, several recent works have used the
EM emanations generated by a monitored system
as a consequence of a certain execution within it, to detect malware
in IoT devices. An example of this is EDDIE \cite{Nazari2017}, a method
that studies the spikes on the EM spectrum generated during a program
execution and compares them to other peaks previously learned during a
training stage. Significant differences in the spikes of EM spectrum allow
to infer the introduction of malware in the studied program. \citet{Han2017}
give a similar approach, presenting ZEUS. This is a contactless
embedded controller security monitor that is able to ensure the integrity
of certain operations, by leveraging the EM emission produced during their
execution, with no additional hardware support or software modifications.
\citet{Yang2019} and \citet{Liu2018} present very similar studies.
The latter additionally considers the problem in terms of participants of an attestation
protocol, namely, verifier and prover systems, although they do not develop a
complete protocol itself. \citet{Msgna2014} give another proposal that uses side-channel signals to check
the integrity of program executions, where the use
of power consumption templates is suitable to verify the integrity of code
without previous knowledge about it.
However, \citet{DBLP:conf/micro/SehatbakhshNKZP19} propose the first attestation
protocol based on EM signals with EMMA. The authors
observe that execution time is only one of the multiple examples
of using measurable side-channel information to gain knowledge about
a specific computation, providing in many cases much finer-grain
information than a unique temporal parameter. To develop this idea,
they design a new attestation method based on the EM
emanations generated by the prover while computing the checksum
challenge proposed by the verifier, instead of the request-to-response
time. After this, they show the implementation of this design, and
consider and evaluate different attacks on EMMA.
On the negative side, and
opposed to classical attestation, this method requires physical proximity
to the setup that captures the emanation, namely, a probe connected to an
oscilloscope or external Software Defined Radio (SDR).

On a separate issue, several recent studies consider the possibility of
attacking cryptosystems by using the side-channel information provided by
mixed-signal components, such as ADCs \cite{DBLP:journals/tches/GnadKT19,
DBLP:journals/tches/OFlynnD19}, or other sensors such as ring oscillators (ROs)
\cite{DBLP:conf/reconfig/GravellierDTL19, DBLP:conf/sp/ZhaoS18, DBLP:conf/fccm/RameshPDPPHT18}
and time-to-digital converters (TDCs) \cite{DBLP:conf/cardis/GravellierDTLO19,
DBLP:journals/dt/GnadKTSM20, DBLP:conf/iccad/SchellenbergG0T18, DBLP:conf/date/SchellenbergG0T18}.
In real-world devices, these components are already placed in the same FPGA or SoC as
a certain cryptographic module that is running some operations with secret parameters,
or available through some interfaces present in processors, e.g.\ Intel Running
Average Power Limit (RAPL) \cite{Lipp2020}. However, to the best of
our knowledge, no study exists aiming to perform
an attestation process by using the side-channel information captured by
said components, leading to what we call \textit{remote power analysis for attestation}.
This technique features the benefits of the EM side channel (finer-grain
information) without the negative sides (external setup with physical proximity).
\section{Remote Power Analysis for Attestation} \label{sec:measure}

\subsection{System description and measurements}

As mentioned in the previous section, the goal of this work is to attest an
operation run in a system by using the power leakage caused by the operation
itself. In this context, where the power consumption traces can be acquired
remotely, they can be used to attest an operation, because the procurement process
can be automated. Generally, the voltage fluctuations caused by the operation
can be captured by any mixed-signal component, that could be an ADC, a sensor
implemented on the programmable logic (PL), or any power supply monitor.

\autoref{fig:scheme} gives an overview of our system. It contains an
AP where the operation (prover) and attestation (verifier) processes are run.
Additionally, it contains the mixed-signal component that measures the supply voltage via ADC with Direct
Memory Access (DMA) while the operation is run.
The verification process saves the power trace captured by the mixed-signal
component as binary data. Since we are using this side-channel trace data
as evidence, the data must be trustworthy. Therefore, in our system, the ADC is inside the trust
perimeter. Exactly how this happens in practice depends on the TEE technology.
For example, on a platform that supports virtualization, this might be
accomplished by a two-stage Memory Mangement Unit (MMU), where the hypervisor
(or TEE) removes access from the untrusted High Level Operating System (HLOS) by
simply not mapping the second-stage translation that would allow access to the
ADC's physical address space or the memory where it stores its data. This is
indeed the scenario that \autoref{fig:scheme} depicts. On ARM-based SoCs this
could also be accomplished with a Memory Protection Unit (MPU) that would be
configured by a TZ-based TEE. The analogous upcoming technology for RISC-V would
be Physical Memory Protection (PMP) \cite{DBLP:conf/eurosys/LeeKSAS20}.
So while the concrete protection mechanism
on a given architecture depends on the TEE implementation, in this work we
generically use the Linux kernel to simulate the TEE in terms of trust, and the
kernel gates all userspace access to the ADC with traditional MMU-based access
control.

\begin{figure}[t!]
\centering
\iftrue
\includegraphics[width=1.0\columnwidth]{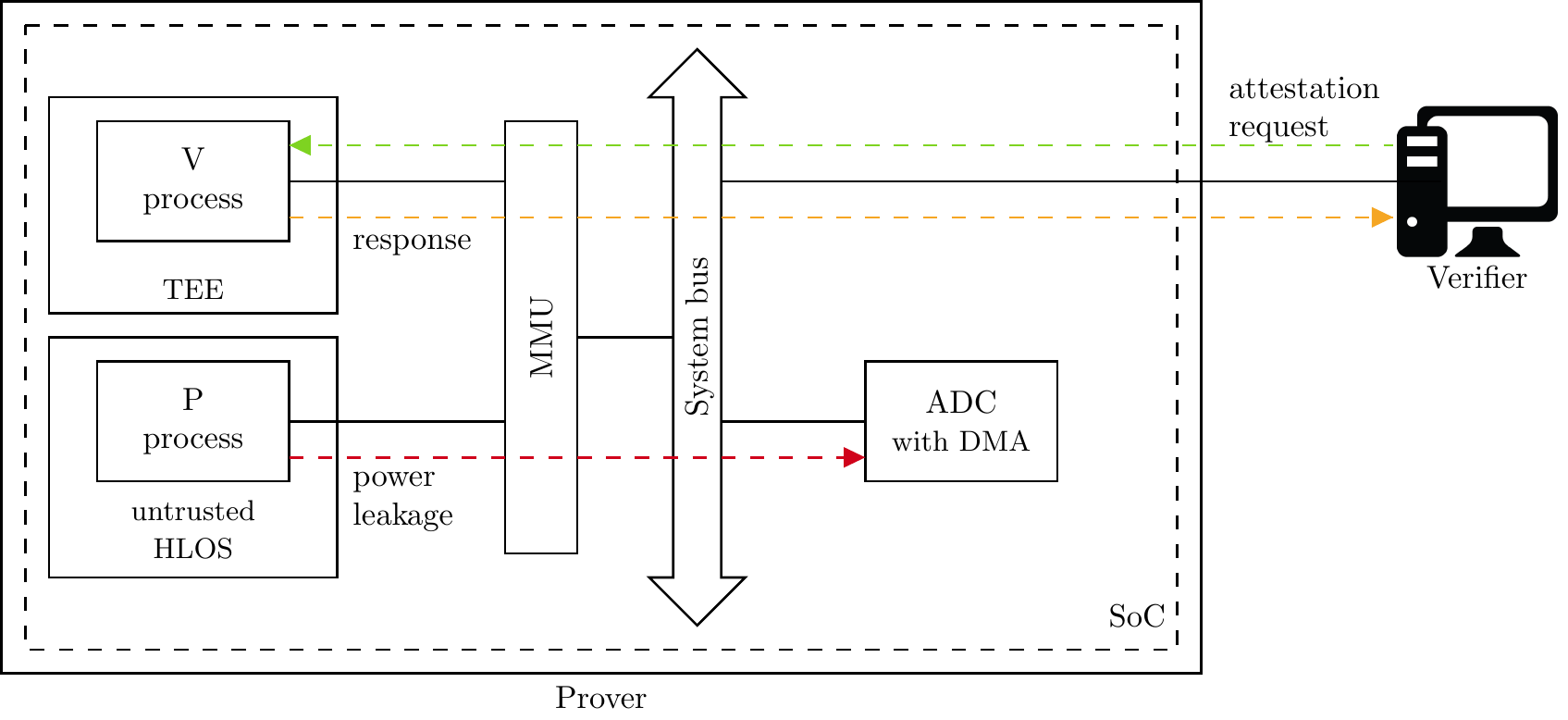}
\else
\Huge\bf PLACE\\HOLDER
\fi
\caption{Block diagram of the attestation process while running the operation.}
\label{fig:scheme}
\end{figure}

Specifically, in this paper we conduct our experiments on
a PYNQ-Z1 board. We programmed the FPGA of its Zynq 7000 chip
to activate the ADC present on it, and capture the voltage fluctuations produced
due to the execution of different operations inside the ARM Cortex-A9 processor.
With this approach, we aim to carry out the attestation process of an execution
performed within the processor, simply with the side-channel information,
given by the voltage fluctuations at the input of the ADC. This is present in
the PL, and totally isolated from the AP core, having only a shared
PDN as a common element.

The XADC module is a hard macro available in the FPGA of the Zynq 7000 chip.
This module is not only an ADC converter of the analog data connected to the
input channel, but it can also be configured to monitor the supply voltages
and temperature. The XADC module supports multichannel, however we configure it
with a single channel that monitors the internal core supply voltage.
The output data size of this module is 16-bit data, with 12-bit precision. The
sample rate is 1MSPS. The 32-bit AXI streaming output of the XADC is used to
transfer as many samples as possible to the AP. The XADC outputs samples on the
AXI Stream for each of its channels when it is enabled, in this case, only one
channel. We use DMA transfer from the PL to the AP to
move XADC samples into the AP memory, and we use an AXI GPIO to set the size of
the transfer. Since the width of the AXI stream is double the output data
size, there are two measurements at each memory position. %

During the attestation process, first the number of power measurements to be
captured is set: that is, the buffer size. Just before running the operation, the
DMA is enabled, then we run a trigger operation to indicate in the signal the beginning of the
operation to be attested, then the operation itself is run, and finally, another
trigger operation indicates that it has ended. The XADC is capturing power data
until the DMA transfer is completed. The AP reads the part of the memory where the
measurements are stored and processes it as needed.

We summarize our procurement process as follows, that saves the power traces as
binary data used in the attestation protocol.
(i) Set the buffer size;
(ii) enable the DMA engine;
(iii) send the start trigger;
(iv) execute the target binary;
(v) send the end trigger;
(vi) wait for the DMA to complete;
(vii) process the resulting binary data (trace).
In practice, the (untrusted) HLOS executes step (iv) and the TEE executes all other steps.

\begin{figure}
\centering
\iftrue
\includegraphics[width=\columnwidth]{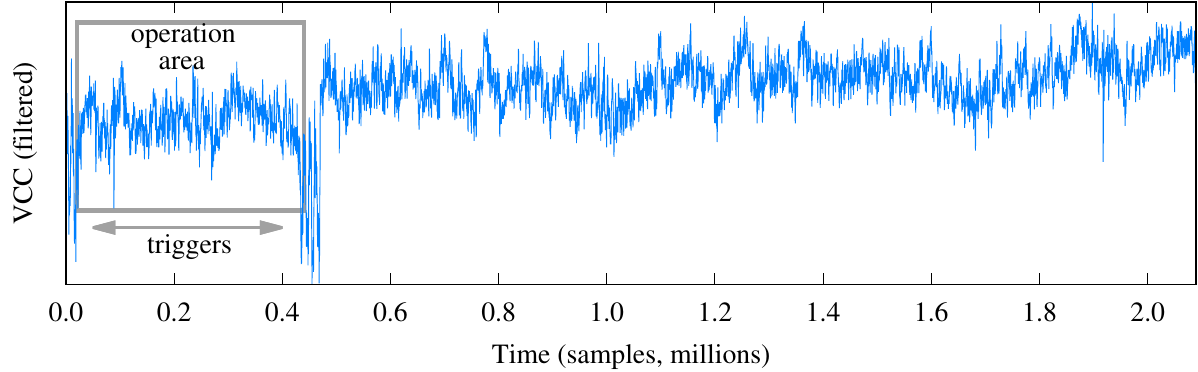}
\else
\Huge\bf PLACE\\HOLDER
\fi
\caption{Example of a power trace captured with the XADC.}
\label{fig:ex_powertrace}
\end{figure}

\autoref{fig:ex_powertrace} shows a power trace captured with the XADC while an
operation is run. The two trigger operations are shown at the beginning and the
end of the operation, determining the operation area. The rest of the trace is the
value of the supply voltage after the operation is finished. This trace is
subsequently processed in the protocol to attest the operation.
\section{System Model and Protocol Construction} \label{sec:protocol}

In this section, we demonstrate how our attestation technique can be
applied and used in real-life deployments and not remain just a
lab-concept. To this end, we present a detailed protocol showing the
communication and all the messages exchanged between the involved
entities. Our protocol description includes the definition of the
underlying system model, as well as the presentation of all the
involved entities and their specifications. Finally, we construct
our threat model and show our protocol's resistance against a
powerful malicious adversary $\mathcal{ADV}$. While this is not the core
contribution of this work, we consider it an important part
since it tackles a problem so far only dealt with at a high
level in other similar works (e.g.~\cite{DBLP:conf/micro/SehatbakhshNKZP19}).
We believe our approach can provide an impetus towards paving the
way for the integration of our, or similar, techniques in existing
services.

\subsection{System Model}
\label{subsec:architecture}
We assume the existence of the following components:

\paragraph*{\textbf{Verifier ($\mathcal{V}$)}:} Here, verifier is a user
who wishes to execute a piece of software on an untrusted platform. Prior
to exchanging the data with the untrusted platform, the user needs
to verify its trustworthiness.

\paragraph*{\textbf{Prover ($\mathcal{P}$)}:} The prover is an untrusted
platform that needs to convince a verifier of its trustworthiness. It consists
of an untrusted application and a TEE.

\begin{enumerate}
  \item \underline{Untrusted application}: The application handles the
  communication between the verifier and the untrusted platform. After proving
  its trustworthiness, the application will be responsible for executing
  software specified by the user.
  \item \underline{Trusted Execution Environment}: We assume the existence of a
  TEE residing either on the untrusted platform or in a remote location. The
  TEE is invoked by the untrusted application upon receiving an attestation
  request by a verifier. TEE's main responsibility is to measure the power
  consumption of the untrusted part of $\mathcal{P}$, while running an
  application requested by $\mathcal{V}$. (Here we recall
  that the TEE also hosts the component that takes the measurements.)
\end{enumerate}

\paragraph*{\textbf{Measurements Tray (MT)}:} MT is an entity residing in the
cloud. Its main responsibility is to store templates and compare them with
traces that are received by $\mathcal{V}$. There are two separate reasons that
led us to have MT as an independent component and not as a part of
$\mathcal{V}$.
(i) MT residing on $\mathcal{V}$'s side would result in higher local storage
costs, as $\mathcal{V}$ would have to keep a copy of each template locally.
(ii) Assuming that MT is an independent cloud component, all MT updates are
executed centrally. This eliminates the need for separate updates.

\subsection{Attestation Protocol}
\label{subsec:protocol}
Having defined our system model, we can now proceed to describe our
attestation protocol. Our construction is
divided into three phases: the \textit{Setup Phase}, the
\textit{Trusted Launch Phase} (\autoref{fig:TL}), and the
\textit{Computations Phase} (\autoref{fig:attestation}). For the rest of
this paper, we assume the existence
\cite{DBLP:journals/siamcomp/GoldwasserMR88} of an IND-CCA2 secure public key
cryptosystem, EUF-CMA secure signature scheme, and a first and second
preimage resistant hash function $H(\cdot)$.

\paragraph*{\textbf{Setup Phase:}} During this phase, each entity receives a
public/private key pair. More specifically:

\begin{itemize}
  \item $(\mathsf{pk}_{\mathcal{V}}, \mathsf{sk}_{\mathcal{V}})$ - Verifier
  $\mathcal{V}$'s public/private key pair.
  \item $(\mathsf{pk}_{\mathcal{P}}, \mathsf{sk}_{\mathcal{P}})$ - Prover
  $\mathcal{P}$'s public/private key pair.
  \item $\mathsf{(pk_{MT}, sk_{MT})}$ - MT's public/private key pair.

\end{itemize}

\paragraph*{\textbf{Trusted Launch Phase:}} In this phase, $\mathcal{V}$ wishes
to launch a TEE on the untrusted platform. The TEE will be responsible for
measuring the power consumption while $\mathcal{P}$
executes applications of $\mathcal{V}$'s choice. To facilitate
$\mathcal{V}$, we assume the existence of a setup function $F_s$,
responsible for setting up the TEE\footnote{The specifications of $F_s$ will
be TEE-dependent. However, it must be designed in such a way that any
manipulation will create a noticeable time increase in the computation of the
checksum.}. Finally, we further assume that the setup function $F_s$ is
publicly known.

This phase commences with the verifier $\mathcal{V}$ generating a random number
$r_1$ and sending $m_1 = \langle r_1, A \rangle$ to $\mathcal{P}$, where
$A$ is the unique identifier of the application that $\mathcal{V}$ wishes to
execute on the TEE. Moreover, $\mathcal{V}$ captures the current time $t_1$.
Upon reception, $\mathcal{P}$ calculates
$\mathrm{checksum}(r_1, F_s)$, and gets the result $\mathsf{res}$.
After the successful execution of $F_s$, a new TEE is launched on the
untrusted platform. Upon its creation, the TEE also obtains a public/private
key pair $\mathsf{(pk_{TEE}, sk_{TEE}})$ (sealing/unsealing keys).
The result of the checksum will be then sent back to $\mathcal{V}$ along
with the launched TEE's public key. More precisely, $\mathcal{P}$ sends the
following message to $\mathcal{V}$:
${m_2 = \langle r_2,  \mathsf{Enc}_{\mathsf{pk}_{\mathcal{V}}}(\mathcal{P},
\mathsf{res, pk_{TEE}} ), \sigma_{\mathcal{P}}(H_1) \rangle}$,
where ${H_1 = H(r_2||\mathcal{P}||\mathsf{res||pk_{TEE}})}$. Upon reception,
$\mathcal{V}$ captures the time $t_2$ and calculates $\Delta t = t_2 - t_1$
(possible because $\mathcal{V}$
also knows the function $F_s$). If $\Delta t$ is as expected, then
$\mathcal{V}$ knows that there is a TEE residing on the untrusted platform.
\autoref{fig:TL} illustrates this phase.

\begin{figure}
  \centering
  \ifLINUXBUILD
  \includegraphics[width=1.0\linewidth, trim=6mm 22mm 0 0, clip]{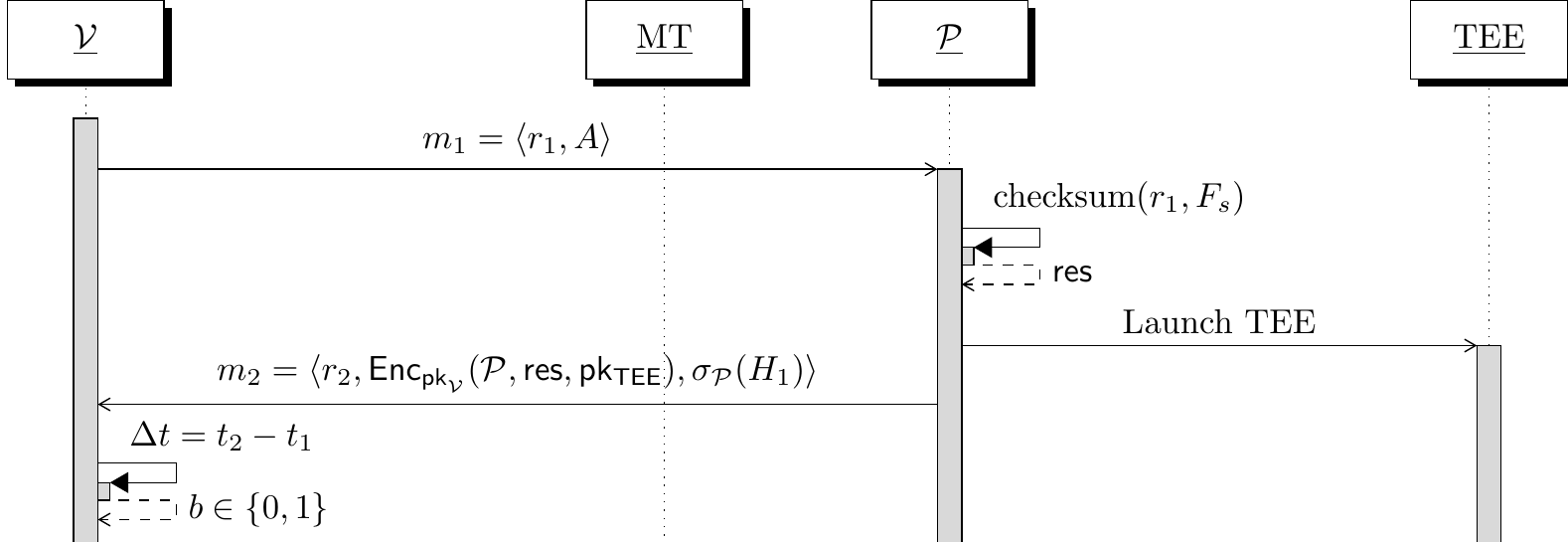}
  \else
  \Huge\bf PLACE\\HOLDER
  \fi
  \caption{Trusted Launch Phase.}
  \label{fig:TL}
\end{figure}

\paragraph*{\textbf{Computations Phase:}} After the successful execution of the
Trusted Launch Phase, $\mathcal{V}$ is convinced that a newly launched TEE is
residing on the untrusted platform. $\mathcal{V}$ wishes to run an executable
application $A$ on the untrusted part of $\mathcal{P}$. To ensure that the
results will be accurate, $\mathcal{V}$ first decides the number of required
traces. This decision depends on the statistical results
described later (\autoref{sec:eval}). After deciding the number of runs $n$,
$\mathcal{V}$ initiates the protocol. To this end,
$\mathcal{V}$ first generates a token $\tau$ and a fresh random number $r_3$,
and contacts $\mathcal{P}$ by sending $m_3 = \langle r_3, n, \tau, A,
\sigma_{\mathcal{V}}(H_2) \rangle$, where $H_2 = H(r_3||\tau||A||n)$ and $A$
is the unique identifier of the application that $\mathcal{V}$ wishes to execute
on $\mathcal{P}$. Upon reception, $\mathcal{P}$ starts running application
$A \ n$ times  with $\tau$ as input, and produces an output $\mathsf{out}$, and
a fingerprint $H(\tau, A)$. Simultaneously, the TEE measures the power
consumption of the untrusted part of $\mathcal{P}$ to get a sequence of
traces $\{tr\}_{i=1}^n$ (one for each execution of $A$).
As soon as $\mathcal{P}$ outputs $\mathsf{out}$,
it sends an acknowledgement $\mathsf{ack}$ to the TEE. Upon reception, the TEE
will
reply to $\mathcal{P}$ with $m_4 = \langle r_4,
\mathsf{Enc}_{pk_{\mathcal{V}}}(\{tr\}_{i=1}^n), \sigma_{TEE}(H_3) \rangle$,
where $H_3 = H(r_4||tr_1||\dots||tr_n)$. $\mathcal{P}$ will finally send $m_5 =
\langle r_5, m_4, H(\tau, A), \mathsf{out}, \sigma_{\mathcal{P}}(H_4)$,
where $H_4 = H(r_5||m_4||\mathsf{out})$. Upon receiving $m_5$,
$\mathcal{V}$ verifies the signatures of the TEE and $\mathcal{P}$, and
the freshness of both $m_4$ and $m_5$ messages. After the first successful
execution of the protocol, $\mathcal{V}$ commences a fresh run, until she
gathers all the required traces. When $\mathcal{V}$ gets the desired number of
traces, she generates $m_6 = \langle r_6,
\mathsf{Enc}_{\mathsf{pk_{MT}}}(\tau, \mathsf{out}, A, \{tr_i\}_{i=1}^n),
\sigma_{\mathcal{V}}(H_3) \rangle$, where $H_3 =
H(r_6 || \tau|| A ||\mathsf{out} || \{tr_i\}_{i=1}^n)$ and sends it to MT.
MT can then check the trust level of $\mathcal{P}$ by comparing $\mathsf{out}$
and each $tr_i$ against its pre-computed list of measurements. Finally, MT
outputs a bit $b \in \{0, 1\}$ and sends $m_7 = \langle r_7,
\mathsf{Enc}_{\mathsf{pk}_{\mathcal{V}}}(b), \sigma_{MT}(H_4) \rangle$,
where $H_4 = H(r_7||b)$ to $\mathcal{V}$.
\autoref{fig:attestation} depicts the Computations Phase.

\begin{figure}
  \centering
  \ifLINUXBUILD
  \includegraphics[width=1.0\linewidth, trim=6mm 16mm 0 0, clip]{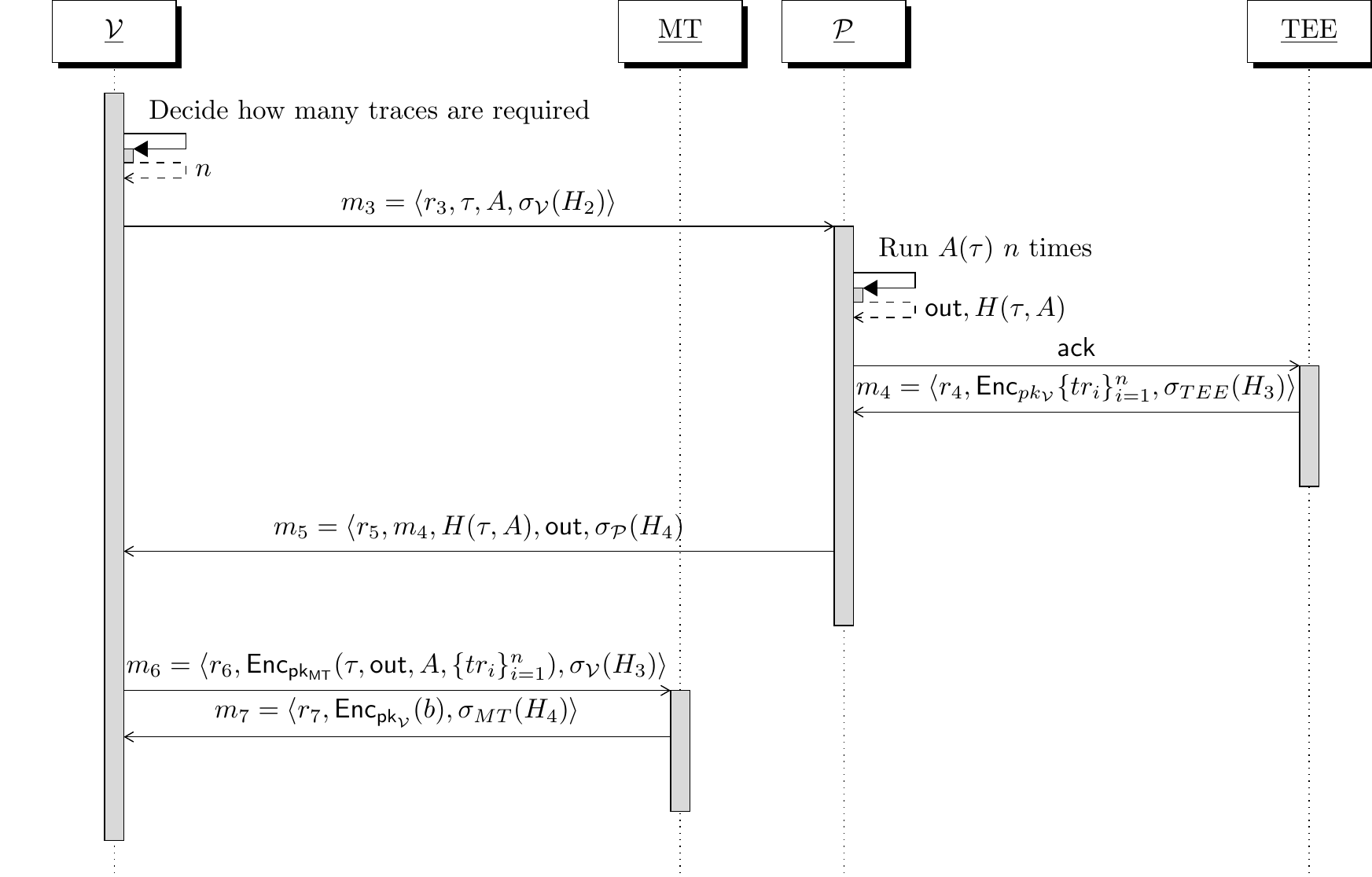}
  \else
  \Huge\bf PLACE\\HOLDER
  \fi
  \caption{Computations Phase: we assume $\mathcal{P}$ and the TEE reside
   on the same platform.}
  \label{fig:attestation}
\end{figure}

\subsection{Threat Model}
\label{sec:ThreatModel}

Our threat model is based on the Dolev-Yao
adversarial model~\cite{dolev:1983}. Furthermore, we assume that
$\mathcal{ADV}$ can load programs of her choice in the enclaves and observe
their output. This assumption significantly strengthens $\mathcal{ADV}$ since
we need to ensure that such an attack will not be detectable from
$\mathcal{V}$'s point of view. Finally, we extend the above threat model by
defining a set of attacks available to $\mathcal{ADV}$.

\begin{myAttack}[Measurements Substitution Attack] Let $\mathcal{ADV}$ be an
adversary that has full control of the untrusted part of $\mathcal{P}$.
$\mathcal{ADV}$ successfully launches a Measurements Substitution Attack if she
manages to substitute the measurements received from TEE by some others of her
choice, in a way that is indistinguishable for $\mathcal{V}$.
\end{myAttack}

\begin{myAttack}[False Result Attack] Let $\mathcal{ADV}$ be an adversary that
overhears the communication between $\mathcal{V}$ and MT. $\mathcal{ADV}$
successfully launches a False Result Attack if she can tamper with the response
sent from MT to $\mathcal{V}$.
\end{myAttack}

While the first two attacks target directly the protocol communication, we also
define a third attack that aims at targeting a false positive (FP) case,
analyzed further in \autoref{sec:eval}. An adversary could exploit the FP by
substituting $\mathcal{V}$'s application with another of her choice. The
resulting trace, even if it comes from a different application, could still
pass as valid by $\mathcal{V}$.

\begin{myAttack}[Application Substitution Attack] \label{attack:application}
Let $\mathcal{ADV}$ be an adversary that overhears the communication between
$\mathcal{V}$ and $\mathcal{P}$. $\mathcal{ADV}$ successfully launches an
Application Substitution Attack if she manages to replace the trace that
$\mathcal{V}$ is expecting with another of her choice, with non-negligible
advantage, where the advantage of $\mathcal{ADV}$ is defined to be the
following conditional probability:
\begin{equation*}
\mathrm{Adv}_{\mathcal{ADV}} = Pr[\mathcal{V}\ \text{accepts the trace}\ |\
\mathcal{ADV}\ \text{switched the application}]
\end{equation*}
\end{myAttack}

\subsection{Security Analysis}
\label{sec:SecAnal}

We now prove the security of our protocol in the presence of a
malicious adversary $\mathcal{ADV}$ as defined in \autoref{sec:ThreatModel}.

\begin{proposition}[Measurements Substitution Attack Soundness] Let
$\mathcal{ADV}$ be an adversary that has full control of the untrusted part of
$\mathcal{P}$. Then $\mathcal{ADV}$ cannot perform a Measurements Substitution
Attack.
\end{proposition}

\begin{proof}
For $\mathcal{ADV}$ to successfully launch a Measurements Substitution Attack,
she needs to replace the measurements $\mu$, with some other measurements
$\mu'$ of her choice. To do so, $\mathcal{ADV}$ can either generate a fresh
$\mu'$, or replay an old one. In both cases, $\mathcal{ADV}$ must
generate a message ${m_5 = \langle r, m_4, H(\tau, A),
\mathsf{out}, \sigma_{\mathcal{P}}(H(r||m_4||\mathsf{out}))\rangle}$. It is
clear from the message structure, that the only component that
$\mathcal{P}$ cannot forge, is the message $m_4$ included in $m_5$.
As $m_4$ is signed by the TEE, and given the EUF-CMA security of
the signature scheme, $\mathcal{ADV}$ can only forge the TEE's signature with
negligible probability. Hence, the only alternative for $\mathcal{ADV}$ is to
use an older $m_4$ message that she received from the TEE sometime in the past.
Let $m_{4_{old}}$ be the old $m_4$ message such that $m_{4_{old}} = \langle
r_{4_{old}}, \mathsf{Enc}_{pk_{\mathcal{V}}}(\mu), \sigma_{TEE}(H_3) \rangle$,
where $H_3 = H(r_{4_{old}}||\mu') $. While this approach solves the problem
of forging TEE's signature, $\mathcal{ADV}$ now needs to further tamper with
this message by replacing $r_{4_{old}}$, with a fresh random number. This is
important because otherwise, $\mathcal{V}$ will not be able to verify the
freshness of the message, and will thus abort the protocol. However,
$r_{4_{old}}$ is included in the signed hash of the TEE, and given the
second preimage resistance of the hash function $H$ we have that $H(r, \mu')
\neq H(r', \mu') \ \forall r, r'$ such that $r \neq r'$. Hence,
$\mathcal{V}$ realizes that something is wrong and aborts the protocol.

\end{proof}

\begin{proposition}[False Result Attack Soundness]  Let $\mathcal{ADV}$ be a
malicious adversary that overhears the communication between $\mathcal{V}$ and
MT. Then $\mathcal{ADV}$ cannot successfully perform a False Result Attack.

\end{proposition}

\begin{proof}
For $\mathcal{ADV}$ to launch a False Result Attack, she needs to forge the
message $m_6$ sent by MT to $\mathcal{V}$ in a way that $\mathcal{V}$ will not
be able to distinguish any difference. To do so, $\mathcal{ADV}$ has two
choices:
(i) substitute the encrypted bit, with a bit of her choice;
(ii) replay an older message.

Substituting the encrypted result is feasible since
$\mathcal{V}$'s public key is publicly known. Hence, it is straightforward for
$\mathcal{ADV}$ to encrypt a bit under $\mathsf{pk}_{\mathcal{V}}$ and replace
it with the actual encrypted result. However, since the encrypted bit is also
included in the signed hash, $\mathcal{V}$ will be able to ascertain that
the integrity of the message has been violated. Thereupon, for $\mathcal{ADV}$
to successfully substitute the encrypted bit, she needs to also forge the MT's
signature. Given the EUF-CMA security of the signature scheme,
this can only happen with negligible probability and so, the attack fails.

Insomuch as $\mathcal{ADV}$ overhears the communication between
$\mathcal{V}$ and MT, she has knowledge of the random numbers used to ensure
the freshness of the messages. On that account, $\mathcal{ADV}$ could try to
forward to $\mathcal{V}$ an older $m_7$ message, with a fresh random number.
However, just like in the previous case, the random number is also included in
the signed hash, and consequently $\mathcal{ADV}$ would once again have to
forge the MT's signature, which can only happen with negligible probability.

The above proofs support our claim that in both cases the attack can only
succeed with negligible probability. As a result, $\mathcal{ADV}$ cannot
successfully launch a False Result Attack.

\end{proof}

\begin{proposition}
Let $n$ be the total number of traces captured to perform an attestation process.
Let $p_{\alpha}$ be the probability of an attacker obtaining
a success result with a single trace derived from another operation,
and $p_{\beta}$ the honest user success probability, using
a single trace coming from the appropriate operation. Assuming that
$p_{\beta}>p_{\alpha}$, there exists a threshold number of traces,
$x_{th}$, required to pass the attestation process, for which
$P(\alpha)=0$ and $P(\beta)=1$,
using a sufficiently large number $n$ of traces.
\label{prop:several_traces}
\end{proposition}

\begin{proof}Let $\{X_{i}\}$ be a succession of independent random
variables that take one of two different results:
\[
X_{i} =
\begin{cases}
0, \text{ do not pass the attestation process} \\
1, \text{ pass the attestation process}
\end{cases}
\]
Let $x=\sum_{i=1}^{n}X_{i}$ be the number of times we pass the attestation process.
We know, by the strong law of large numbers (SLLN) that:
\begin{equation}
P(\lim_{n \to \infty}\frac{\sum_{i=1}^{n}X_{i}}{n}=E(X_{i}))=1
\label{eq:expectation}
\end{equation}
where $E(X_{i})$ is the expected value of variable $X_{i}$. From the
binomial distribution formula:
\begin{equation}
P(x)=\binom{n}{x}\cdot p^{x}\cdot(1-p)^{n-x}
\label{eq:binomial}
\end{equation}
Leveraging that, in the binomial distribution, the expected value
$E(X_{i})$ of a variable matches with its probability $p$. We can
substitute in \autoref{eq:expectation}, yielding:
\[
P(\lim_{n \to \infty}\frac{\sum_{i=1}^{n}X_{i}}{n}=p)=1\Rightarrow
P(\lim_{n\rightarrow\infty}\frac{x}{n}=p)=1
\]
Let us consider now the two different cases of $p_{\alpha}$ and $p_{\beta}$.
For a sufficiently large $n$, we get that:
\begin{equation}
P(\lim_{n \to \infty}\frac{x_{\alpha}}{n}=p_{\alpha})=1\Rightarrow \lim_{n \to
\infty}\frac{x_{\alpha}}{n}=p_{\alpha} \text{ almost surely}
\footnote{Notice that ``almost surely'' is a concept used in probability theory
to describe events that occur with probability 1 when the sample space
is an infinite set.}
\label{eq:lim_p_alpha}
\end{equation}

\begin{equation}
P(\lim_{n \to \infty}\frac{x_{\beta}}{n}=p_{\beta})=1\Rightarrow \lim_{n \to
\infty}\frac{x_{\beta}}{n}=p_{\beta} \text{ almost surely}
\label{eq:lim_p_beta}
\end{equation}
Since $p_{\beta}>p_{\alpha}$ by assumption
(an essential condition to perform a solid attestation), we
can select $x_{th}$ defined as the threshold number of traces such
that $p_{\alpha}<\frac{x_{th}}{n}<p_{\beta}$ for which we need a
number $\frac{x}{n}\geq\frac{x_{th}}{n}$ to have a positive result
of the attestation process.
Then, from \autoref{eq:lim_p_alpha} and the fact that $p_{\alpha}<\frac{x_{th}}{n}$
by definition, we get:
\[
P(\alpha)=P(\lim_{n \to \infty}\frac{x_{\alpha}}{n}\geq\frac{x_{th}}{n})=0
\]
Analogously, from \autoref{eq:lim_p_beta} and the fact
that $p_{\beta}>\frac{x_{th}}{n}$ by definition, we get:
\[
P(\beta)=P(\lim_{n \to \infty}\frac{x_{\beta}}{n}\geq\frac{x_{th}}{n})=1
\]
\end{proof}
\section{Evaluation} \label{sec:eval}
With the previous sections explaining how we capture the traces from the ADC and how we establish
security between the prover and the verifier system, we are in a
position to explain our experiments.
We open with a description of our analysis framework, including empirical evaluation (\autoref{subsec:corr}).
We then close with a security analysis of the different framework parameters (\autoref{subsec:several_traces}),
guiding selection when instantiating the \autoref{sec:protocol} protocols.

\subsection{Methodology} \label{subsec:corr}
Our procedure consists of
(i) selecting the traces generated by the ADC that belong to a given program;
(ii) generating a template by averaging a large number of traces; and
(iii) comparing this template to different traces, some that belong and some
that do not to the same given program.
Our methodology uses profiling both to build the templates and calculate the
correlation threshold that operations should surpass to complete the attestation,
both which vary across binaries (see Table 2). As an alternative,
non-profiled approaches could be an interesting research direction to potentially
improve scalability and agility.
We utilize the Pearson correlation for our comparison metric.
These, ultimately, will lead to
statistics about the true positive (TP), true negative (TN), false negative (FN),
and false positive (FP) rates that will allow us to make a concrete analysis
concerning the suitability of the traces retrieved by a given sensor or
peripheral to perform attestation. From this data, we obtain
the following parameters, which are typical in information classification,
that give an idea of the accuracy and relevance of our experiments.
\begin{equation*}
Precision = \frac{TP}{TP+FP}
\quad
Recall = \frac{TP}{TP+FN}
\quad
F1 = 2 \cdot \frac{Recall \cdot Precision}{Recall+Precision}
\end{equation*}
This way, precision gives a measure of the number of correct results among
all the returned results, while recall gives a measure of the number of
correct results divided by the number of results that should have been returned.
This means that a low result of precision implies that a high number of incorrect
results are considered as correct, so our system would be yielding many FP. On
the other hand, a low recall implies that we are not considering as correct
some results that indeed are correct, leading then, to a high number of FN.
Finally, F1 is the harmonic mean between recall and precision, and allows us
to give an idea of how good our system is at retrieving results with one
single measurement.

Specifically, we utilize the executables in the Bristol Energy Efficiency Benchmark
Suite (BEEBS) \cite{DBLP:journals/corr/PallisterHB13,DBLP:journals/cj/PallisterHB15},
providing a broad spectrum of programs\footnote{\url{https://github.com/mageec/beebs}}
to profile w.r.t.\ our methodology. We execute the BEEBS programs and capture
their traces from the ADC in order to perform attestation by comparing
each trace with templates previously obtained for the other BEEBS programs.
The aim is that a trace coming from a certain program only
matches (leading to a high correlation value) with the template belonging to its
own program and does not match (yielding a low correlation value) with
other programs' templates.

To accomplish this, we first capture 1000 traces from each program and
generate templates from them. Then, we apply a Savitzky-Golay filter to
obtain the final template. \autoref{fig:Templates} depicts the final
template of five BEEBS programs, where the trigger operations determine each program operation
area (see \autoref{fig:ex_powertrace}). It is important to notice that the amplitude
shift in the traces is not a reliable differentiator, since it depends on
the moment the traces are taken.

\begin{figure}[t!]
\centering
\iftrue
\includegraphics[width=\columnwidth]{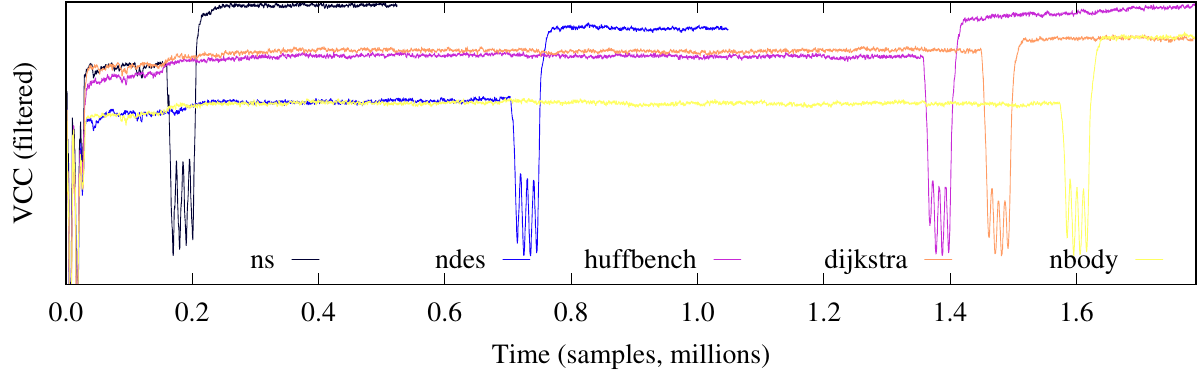}
\else
\Huge\bf PLACE\\HOLDER
\fi
\caption{Several templates selected, with different number of samples.}
\label{fig:Templates}
\end{figure}

One observation from \autoref{fig:Templates} is that the templates associated to
each program take different times, translated into different number of samples,
to complete their execution. The program will be running while the ADC is
capturing samples between the start and the ending trigger. The rest of
the trace after the end trigger is simply noise. Our ADC captures $2^{21}$ samples for
every trace and template, but in order to perform the correlations, each program
and its template are separated into different groups according to their lengths
ranging from $2^{17}$ to $2^{21}$ samples. This way, depending on the
length of the different templates, we keep the execution part from the traces
and templates by selecting a number of samples between
$2^{17}$ to $2^{21}$ samples, trying to catch the relevant information from said
execution. Since the rest of the trace, once the program is
executed and finished, is simply noise, we discard it.

\begin{figure}[t!]
\centering
\iftrue
\includegraphics[width=\columnwidth]{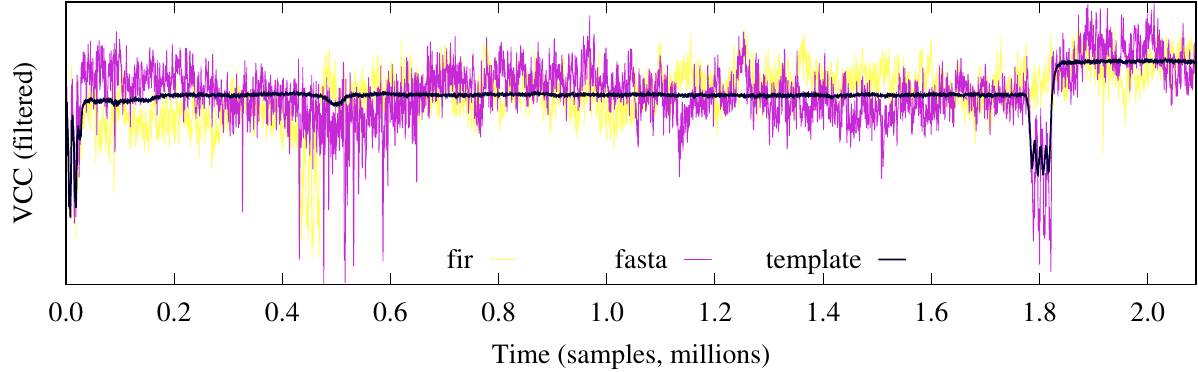}
\else
\Huge\bf PLACE\\HOLDER
\fi
\caption{Matching of the fasta template (black) with two different traces (fasta, fir).}
\label{fig:matching}
\end{figure}

After computing all the templates, our experiment consists of capturing
1000 traces for every program and comparing them against their own template. To
achieve this, we compute the Pearson correlation
between every trace of the program selected and its template.
We store this in
a correlation vector $corr_{vector}= (corr_{1}, corr_{2},..., corr_{1000})$,
then the components of $corr_{vector}$ are ordered from smallest to largest,
in order to compute the 25th percentile. This is to say, we compute a
threshold value $corr_{thres}$ from which 750 components of the vector are above
it. In practical terms, this means that during the matching stage, a trace that
has a correlation value above the threshold will be considered as
representative of an honest execution from a certain program, while a trace with
a correlation value lower than this threshold cannot certify that the trace
belongs to the program related to the template. This matching stage,
in practical terms, works as a training set for our correlation system where we select
a threshold value that let pass the 75\% of the traces from that matching set. When
we move to the evaluation stage, the threshold value is the one that we previously selected,
but it does not need to pass exactly the 75\% of the traces, since the traces from
the evaluation stage are not the same as the one from the matching stage. Nonetheless, it should yield
a similar ratio of traces that pass the attestation. This way, we ensure that the threshold
admits a sufficient number of traces, without allowing a large number coming from other operations.
In \autoref{fig:matching}, it is
visible to the naked-eye how a correct trace matches with its template, against
a trace that does not belong to the program according to the used template.

After this, we select a template and compute correlations for all the BEEBS executables,
with a set of 1000 traces from each program. If traces that do not belong to the program being
checked match with the corresponding template, having a correlation value
higher than the threshold, we will have a FP. On the contrary, if a trace that belongs to this program
does not match with its own template, having a correlation value lower than the threshold, we will have
a FN. From these statistics, we compute recall, precision, and F1 score as previously defined.
To obtain those statistics, we first trim the traces to have the
same number of samples than the template used in each case, in order to be able to
compute the Pearson correlation value between the traces and the given template.
\autoref{fig:explanation} illustrates the whole process, repeated for the templates obtained from each program.
We verified that this process can be carried out even if the templates are generated from a set of data from a
given board of a certain model, and
the evaluation stage where we compute the correlation of the traces with the templates previously generated are
captured from a different board of the same model. This characteristic demonstrates the robustness of our attestation protocol, and its
utility in systems where we can have pre-loaded templates for given binaries coming from different
devices.

\begin{figure} [h!]
\centering
\iftrue
\includegraphics[width=1.0\columnwidth]{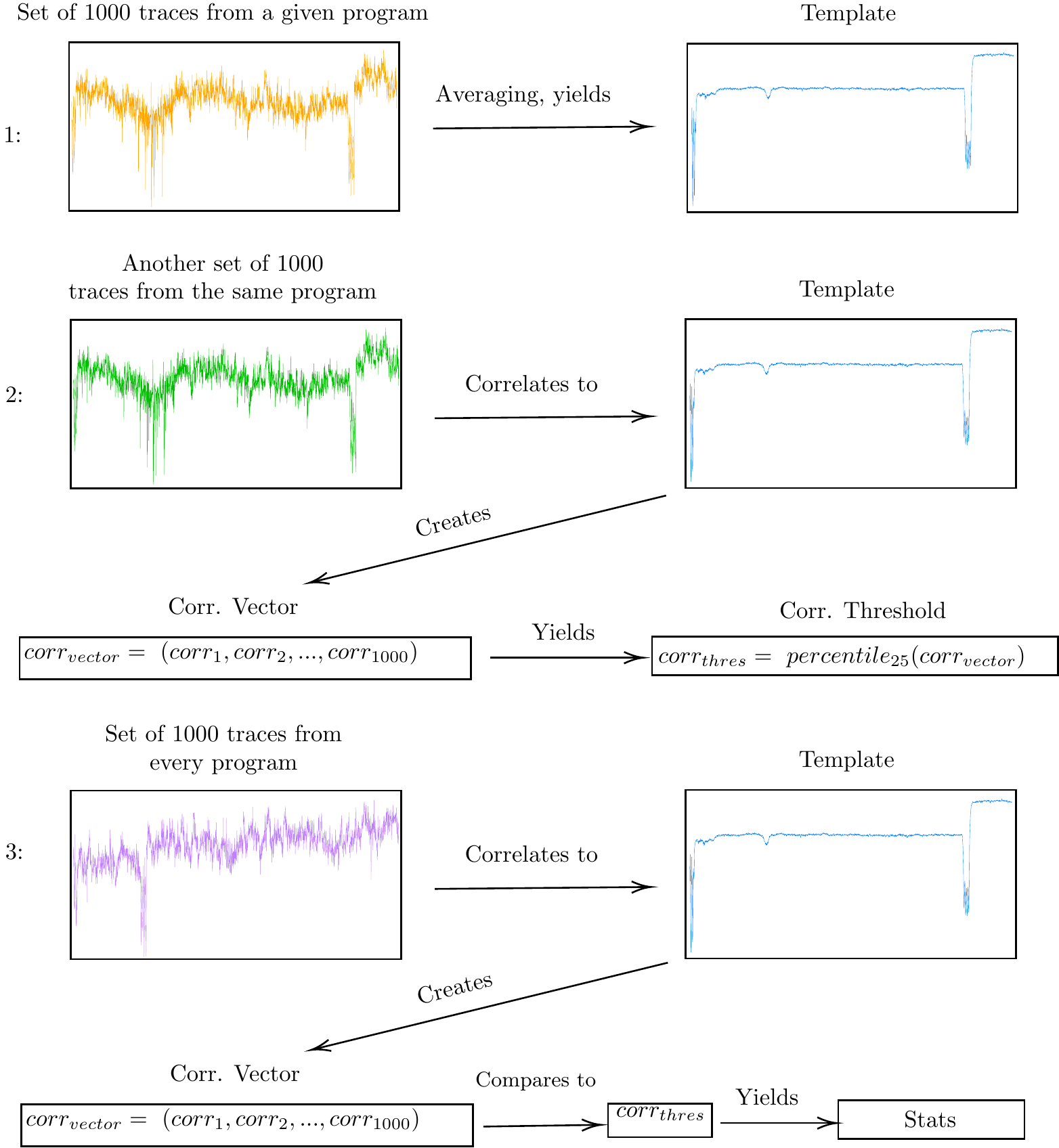}
\else
\Huge\bf PLACE\\HOLDER
\fi
\caption{Process repeated with every template to obtain statistics.}
\label{fig:explanation}
\end{figure}

We present the complete statistics in \autoref{table:CompleteStats}.
It includes the correlation threshold value for each program,
the name of the program that yields the highest FP rate, the absolute number
out of the 1000 traces that provides a correlation value higher than the threshold (leading to FP),
recall, precision, and F1 score metrics. From \autoref{table:CompleteStats}, we can
see that out of 47 programs, 35 have a precision above 0.90 and 42 above 0.80,
while 27 programs have a recall above 0.70 and 44 above 0.60. Thus, our method
correctly distinguishes positive results among the total number of traces, having in
general a low FPR, and leading to excellent precision values. On the contrary,
the recall results are lackluster, meaning that a significant portion of correct results
are not retrieved by our system, leading to an improvable FN rate (FNR).

Another concern is the fact that out of 1000 traces,
in the worst FP case scenario from a specific benchmark, 82 traces from ndes program
are considered as correct results by the nettle-aes template, leading to a FPR of 8.2 \%
from that specific operation. This means that traces coming from a given operation are
likely useful to attest others. The consequences, at this point, are clear. To have said 8.2 \% of FP results, would lead
to an inadmissible number of incorrect traces passing the attestation process. This is especially worrying, taking
into account that the recall (which is equivalent to TPR, the TP rate) is 69 \%.
With these results, each time we want to use this attestation process, we have roughly
a 3/10 probability of failing even if the computations are carried out properly, and
around a 1/12 probability of having a correct result in cases where the retrieved trace matches
with the template, even if it does not belong to the correct program. This is far from
the aimed numbers to perform solid attestation using our proposal.

\subsection{Parameterization} \label{subsec:several_traces}
To overcome these deficiencies,
we can use several traces to perform the attestation. We aim to improve these numbers by
requiring that out of $n$ traces, a minimum number $x_{th}$ must be above the threshold
correlation value. \autoref{prop:several_traces} presents the general case,
but now we must assign discrete values to the various parameters.

For our case, and to obtain a good trade-off between the probabilities
of having FPs that pass the attestation, and having FNs that are not
able to pass even if they belong to the correct template,
we consider the following threshold. It is the midpoint of $p_{\alpha}$ and
$p_{\beta}$, recalling those parameters from \autoref{prop:several_traces}.
\[
x_{th},n\in\mathbb{N}:x_{th}=\left\lceil n \cdot \frac{p_{\alpha}+p_{\beta}}{2}\right\rceil\wedge p_{\alpha}<\frac{x_{th}}{n}<p_{\beta}
\]
It is important to notice that each user can freely select this threshold number of traces by, for example,
giving different weights to $p_{\alpha}$ and $p_{\beta}$ or selecting a totally different relation. As previously
indicated, we select this threshold to have a good trade-off between an attacker trying to cheat the attestation
process with a trace coming from another operation and an honest user.

Moreover, the user determines the security level by selecting an appropriate
number of total traces $n$. Here, we understand the security level as a measure
of the probability of passing the attestation protocol using traces coming from a different
operation, given the number of traces required to be above the correlation threshold and
the total number of traces $n$.
Larger values of $n$ tend to minimize the FP probabilities and
increase the corresponding TP probabilities (thus, reducing also FNs). For our worst FP
scenario coming from a specific benchmark, we have 82 traces out of 1000 coming
from the ndes benchmark that yield FP results using the nettle-aes template. On the other
hand, the TPR for nettle-aes is 0.69. In this case, we can identify
the FPR from the ndes traces, with an attacker's probability to successfully cheat the attestation process of the
nettle-aes operation, thus, $p_{\alpha}=0.082$.
Analogously, the TPR for
nettle-aes is the success probability of an honest user,
thus $p_{\beta}=0.69$. Iterating, we found that for $n=243$ traces, the threshold
needed to pass the attestation is:
\[
x_{th}=\left\lceil243 \cdot \frac{0.69+0.082}{2}\right\rceil=94
\]
Combining our results for the obtained threshold with the binomial distribution (\autoref{eq:binomial})
for $n = 243$ and $p_{\alpha} = 0.082$, we get that the probability to cheat the attestation is:
\[
P(\alpha)=P(x\geq94)=P(94)+P(95)+...+P(243) \text{ using } p_{\alpha}=0.082
\]
\[
P(\alpha)=3.72\cdot10^{-39}\approx\frac{1}{2^{128}}=2.94\cdot10^{-39}
\]
which is equivalent to a 128-bit security level. Using only 243 traces to complete the
attestation process is a very promising result,
especially since this is only a proof-of-concept of our novel attestation approach.

Considering the TPR, for $n=243$ traces and $p_{\beta}= 0.69$, we get that:
\[
P(\beta)=P(x\geq94)=P(94)+P(95)+...+P(243) \text{ using } p_{\beta}=0.69
\]
\[
P(\beta)=1-6.27\cdot10^{-23}\approx 1 - \frac{1}{2^{74}}
\]
Thus, we achieve an (almost) perfect TPR with the selected number of traces. It is important to notice that these 243 traces
are needed in our worst case scenario. Concretely, our worst case scenario is that in which a higher number of traces
coming from a different operation passes the attestation process of another binary.
For the rest of binaries, the security level achieved will be equal or better than this one.
\autoref {tab:security-level} show similar results about the number $n$ of traces and the threshold $x_{th}$ required to pass the attestation
with various security levels for this worst case scenario.

\begin{table}[t!]
\centering
\caption{Parameter variations to achieve different security levels.}
\label{tab:security-level}
\begin{tabular}{|c|c|c|c|c|}
\hline
\phantom{.}Security Level\phantom{.} & $n$ & $x_{th}$ & $P(\alpha)$                                    & $P(\beta)$                \\ \hline
32-bit  \vphantom{$0^{0^0}$} & 52  & 21       & $2.39 \cdot 10^{-10} \approx 2^{-32}$  & $1 - 5.43 \cdot 10^{-6} \approx 1 - 2^{-17}$  \\ \hline
64-bit  \vphantom{$0^{0^0}$} & 114 & 45       & $5.18 \cdot 10^{-20} \approx 2^{-64}$  & $1 - 2.22 \cdot 10^{-11} \approx 1 - 2^{-35}$ \\ \hline
128-bit \vphantom{$0^{0^0}$} & 243 & 94       & $3.72 \cdot 10^{-39} \approx 2^{-128}$ & $1 - 6.27 \cdot 10^{-23} \approx 1 - 2^{-74}$ \\ \hline
256-bit \vphantom{$0^{0^0}$} & \phantom{.}494\phantom{.}
                             & \phantom{.}191\phantom{.}
                             & \phantom{.}$9.83 \cdot 10^{-78} \approx 2^{-256}$\phantom{.}
                             & \phantom{.}$1-4.14 \cdot 10^{-44} \approx 1 - 2^{-144}$\phantom{.} \\ \hline
\end{tabular}
\end{table}

Ideally, we would achieve a certain security level by isolating the minimum number $n$ of traces
and the $x_{th}$ required to find a given $P(\alpha)$. However, this is not possible, since the inverse function
of the binomial cumulative distribution does not exist. In other words, there is not an analytical form to find
$n$ starting from a given $P(x_{th} \geq x)$. That is the reason why we use the iteration approach, through
readily-available numerical methods that are easily and fastly computed by any tool or programming language
considering the equation from the binomial distribution.

\section{Conclusion} \label{sec:conclusion}
The main contribution of this paper is the proposal
of a new method to verify the integrity of a SoC, by
natively capturing the side-channel leakage produced during
the execution of a given operation.
In our case, an
ADC present in an FPGA adjacent to the AP carrying out
the operation that aims to be verified is suitable to
measure the voltage fluctuations caused by the execution
itself. These voltage fluctuations are able to characterize the performed
operation and to distinguish it from other binaries.

This attestation method does not rely on the
request-to-response time, using instead the power signal
generated by the program execution, which provides
more detailed information about what its proper
behavior should be.
Additionally, our method does not require an external setup
with physical proximity to capture the side-channel vector, rather native components.
Thus, it implies no software or hardware overhead to the system, since
it simply internally captures a power trace while carrying out
executions in a normal operating mode.

To end, our attestation protocol completes the work.
It describes not only how our system can capture the power leakage
that allows us to characterize a given operation, but also how to
manage the resultant power trace to, realistically, verify the
integrity of an untrusted system.
This achieves our main goal: checking that an untrusted system is executing
programs honestly, without the presence of any malware.
As far as we are aware, our work is the first constructive application of
remote power analysis, identified as an open problem in \cite{DBLP:conf/IEEEares/Martinez-Rodriguez21}.

\Paragraph{Limitations and future work}
Our proof-of-concept work exhibits a variety of limitations that should be
addressed in future related studies.
A brief summary follows.
(i) TP rates are improvable, especially taking into account that our traces
are fairly noisy. Nonetheless, using several power traces we are able to
overwhelmingly detect honest users vs.\ attackers.
(ii) Substitution attacks are a real threat, in case the ADC resolution is not
sufficient to capture malicious modifications in the instructions from a binary.
Future work includes exploring attacker strategies to modify binaries to produce
power traces that pass attestation.
(iii) The matching between traces and templates are mainly based
in the duration of the executed operation, since we use the ending triggers
as a distinctive mark in the power traces. However, the fact of using a whole
power trace (vectorial data) instead of the classical request-to-response time
(scalar data) hardens proxy attacks, because an attacker does not only need to
solve a challenge in a given time, but to generate a power trace similar to the
template, which is a difficult challenge.
(iv) Consulting \autoref{fig:scheme}, a natural observation is that requiring a TEE
for dynamic attestation seems paradoxical, in the sense that the target binary
could simply be part of the TEE itself. However, a major goal of TEEs is
reducing the Trusted Computing Base (TCB); keeping the target binaries outside
the immediate TCB significantly narrows the attack surface.

\Paragraph{Acknowledgments}
(i) This project has received funding from the European Research Council (ERC) under
the European Union's Horizon 2020 research and innovation programme (grant
agreement No.\ 804476).
(ii) This project has received funding by the ASCLEPIOS: Advanced Secure Cloud
Encrypted Platform for Internationally Orchestrated Solutions in Healthcare
Project No.\ 826093 EU research project.
(iii) Supported in part by the Cybersecurity Research Award granted by
the Technology Innovation Institute (TII).
(iv) Supported in part by CSIC's i-LINK+ 2019 ``Advancing in cybersecurity technologies''
(Ref.\ LINKA20216).
(v) The first author was financially supported in part by HPY Research Foundation.
(vi) M.~C.~Mart\'inez-Rodr\'iguez holds a postdoc that is co-funded by European Social Fund (ESF) and the Andalusian government, through the Andalucia ESF Operational Programme 2014--2020.

\begin{table}
\footnotesize
\centering
\caption{Complete statistics for every BEEBS benchmark used.}
\label{table:CompleteStats}
\begin{tabular}{@{}|c|c|c|c|c|c|@{}}
\hline
\textbf{Benchmark}            &  \textbf{$corr_{thres}$} & \textbf{Max.\ FP} & \textbf{Precision} & \phantom{Z}\textbf{Recall}\phantom{Z} & \phantom{ZZ} \textbf{F1} \phantom{ZZ} \\
\hline
aha-compress         & 0.7248         & crc32 (23)               & 0.9617    & 0.7540 & 0.8453 \\
bs                   & 0.7096         & newlib-sqrt (37)         & 0.9135    & 0.7710 & 0.8362 \\
bubblesort           & 0.5312         & nbody (30)               & 0.8959    & 0.7490 & 0.8159 \\
cnt                  & 0.7112         & frac (3)                 & 0.9775    & 0.6950 & 0.8124 \\
cover                & 0.7264         & crc (35)                 & 0.9122    & 0.5820 & 0.7106 \\
crc32                & 0.7447         & aha-compress (9)         & 0.9783    & 0.5410 & 0.6967 \\
crc                  & 0.7724         & duff (4)                 & 0.9947    & 0.7570 & 0.8597 \\
ctl-stack            & 0.6795         & sqrt (50)                & 0.6018    & 0.8070 & 0.6895 \\
ctl-vector           & 0.7070          & ns (10)                  & 0.9868    & 0.8200 & 0.8957 \\
cubic                & 0.7232         & sqrt (1)                 & 0.9987    & 0.7890 & 0.8816 \\
dijkstra             & 0.5141         & nettle-des (1)           & 0.9986    & 0.7020 & 0.8244 \\
duff                 & 0.7318         & crc (44)                 & 0.9481    & 0.8040 & 0.8701 \\
fasta                & 0.4935         & bs (12)                  & 0.9374    & 0.7340 & 0.8233 \\
fibcall              & 0.6551         & newlib-log (73)          & 0.6986    & 0.7370 & 0.7173 \\
fir                  & 0.6268         & none                      & 1.0000    & 0.8010 & 0.8895 \\
frac                 & 0.7856         & crc32, newlib-sqrt (2)   & 0.9880    & 0.7420 & 0.8475 \\
huffbench            & 0.5608         & newlib-log (14)          & 0.8525    & 0.6360 & 0.7285 \\
janne\_complex        & 0.4192         & crc (19)                 & 0.8014    & 0.6860 & 0.7392 \\
jfdctint             & 0.7598         & nettle-des (48)          & 0.9360    & 0.7020 & 0.8023 \\
lcdnum               & 0.3827         & cnt (38)                 & 0.6283    & 0.7100 & 0.6667 \\
levenshtein          & 0.5909         & template (4)             & 0.9904    & 0.7250 & 0.8372 \\
matmult-float        & 0.7009         & none                     & 1.0000    & 0.6520 & 0.7893 \\
matmult-int          & 0.4645         & several operations (2)   & 0.9804    & 0.7520 & 0.8511 \\
mergesort            & 0.7296         & sglib-hashtable (44)     & 0.9058    & 0.6730 & 0.7722 \\
nbody                & 0.5610          & bubblesort (30)          & 0.8066    & 0.7130 & 0.7569 \\
ndes                 & 0.6663         & nettle-aes (80)          & 0.8638    & 0.6340 & 0.7313 \\
nettle-aes           & 0.6318         & ndes (82)                & 0.8903    & 0.6900 & 0.7775 \\
nettle-des           & 0.7444         & cover (21)               & 0.9613    & 0.7690 & 0.8545 \\
newlib-log           & 0.5017         & dijsktra (78)            & 0.6399    & 0.6290 & 0.6344 \\
newlib-sqrt          & 0.5191         & dijkstra (81)            & 0.4573    & 0.7170 & 0.5584 \\
ns                   & 0.7599         & none                     & 1.0000    & 0.7320 & 0.8453 \\
nsichneu             & 0.7869         & none                     & 1.0000    & 0.7210 & 0.8379 \\
picojpeg             & 0.6009         & st (13)                  & 0.9805    & 0.6540 & 0.7846 \\
qrduino              & 0.5639         & sglib-listsort (6)       & 0.9834    & 0.4740 & 0.6397 \\
rijndael             & 0.6470          & ndes (5)                 & 0.9927    & 0.6800 & 0.8071 \\
select               & 0.6120          & template (27)            & 0.9415    & 0.7880 & 0.8579 \\
sglib-arrayheapsort  & 0.6372         & bs (2)                   & 0.9912    & 0.6720 & 0.8010 \\
sglib-arrayquicksort & 0.7020          & mergesort (11)           & 0.9807    & 0.7110 & 0.8244 \\
sglib-hashtable      & 0.6919         & mergesort (77)           & 0.8894    & 0.6350 & 0.7410 \\
sglib-listinsertsort & 0.6251         & various (1)              & 0.9947    & 0.7540 & 0.8578 \\
sglib-listsort       & 0.5108         & qrduino (24)             & 0.9492    & 0.6540 & 0.7744 \\
sqrt                 & 0.5113         & rijndael (2)             & 0.9923    & 0.6470 & 0.7833 \\
st                   & 0.5812         & picojpeg (24)            & 0.9644    & 0.6770 & 0.7955 \\
stb\_perlin           & 0.6476         & template (3)             & 0.9833    & 0.6460 & 0.7797 \\
tarai                & 0.6927         & wikisort (37)            & 0.9025    & 0.7500 & 0.8192 \\
template             & 0.6654         & select (19)              & 0.9487    & 0.6840 & 0.7949 \\
wikisort             & 0.7092         & sglib-listinsertsort (9) & 0.9673    & 0.7090 & 0.8182 \\
\hline
\end{tabular}
\end{table}
 \bibliographystyle{splncsnat}

\end{document}